\documentclass[11pt]{article}
\usepackage[margin=0.95in]{geometry}
\usepackage{amsmath,amssymb,amsthm,mathtools,bm,booktabs,array}
\usepackage[T1]{fontenc}
\usepackage[utf8]{inputenc}
\usepackage{microtype}
\usepackage{mathrsfs}
\usepackage[hidelinks]{hyperref}
\newcommand{\ii}{\mathrm i}
\newcommand{\x}{\bm x}
\newcommand{\p}{\bm p}
\newcommand{\Lop}{\bm L}
\newcommand{\J}{\bm J}
\newcommand{\Sig}{\bm\Sigma}
\newcommand{\Del}{\bm\Delta}
\newcommand{\Om}{\bm\Omega}
\newcommand{\sone}{\bm\sigma_1}
\newcommand{\stwo}{\bm\sigma_2}
\newcommand{\id}{\mathbb I}
\newtheorem{theorem}{Theorem}[section]
\newtheorem{proposition}[theorem]{Proposition}

\theoremstyle{remark}
\title{Nested polynomial algebras and the exact spectrum\\of a three-dimensional two-spin oscillator}
\author{%
Ian Marquette\textsuperscript{a}\thanks{%
E-mail: \texttt{i.marquette@latrobe.edu.au}},
O. O\u{g}ulcan Tuncer\textsuperscript{b}\thanks{%
E-mail: \texttt{otuncer@hacettepe.edu.tr}}
 and \.{I}smet Yurdu\c{s}en\textsuperscript{b}\thanks{%
E-mail: \texttt{yurdusen@hacettepe.edu.tr}}\\[0.8em]
\small \textsuperscript{a}Department of Mathematical and Physical Sciences,
La Trobe University,\\
\small Bendigo, VIC 3552, Australia\\[0.4em]
\small \textsuperscript{b}Department of Mathematics,
Hacettepe University,\\
\small 06800 Beytepe, Ankara, Turkey}
\date{}
\hypersetup{pdfauthor={Ian Marquette, O. Ogulcan Tuncer and Ismet Yurdusen},pdftitle={Nested polynomial algebras and the exact spectrum of a three-dimensional two-spin oscillator}}
\begin{document}
\maketitle
\begin{abstract}
We construct representations of the symmetry algebra and derive the exact spectrum of a three-dimensional oscillator with two spin-$\tfrac12$ degrees of freedom. The Hamiltonian and its first-order scalar and pseudo-scalar integrals were obtained in earlier classifications. Here we determine their action on the complete spectral basis. At fixed total angular momentum, a quadratic algebra acts on the three triplet orbital channels and admits a three-dimensional deformed-oscillator representation, while an overlapping two-state exchange algebra joins the middle triplet to the singlet. These finite modules determine angular--spin multiplicities but do not quantize the radial energy. A separate positive-discrete-series $\mathfrak{su}(1,1)$ representation of the common radial Hamiltonian generates the energy sequence. We give the finite matrices, Casimir and structure functions, normalized eigenfunctions, radial ladder actions, the exceptional $j=0$ sector, complete energy-shell multiplicities, and the effect of particle interchange.
\end{abstract}

\section{Introduction}

Algebraic methods can be used to determine energy spectra and describe degeneracies in quantum mechanics. Energy-preserving operators act within degenerate eigenspaces, whereas spectrum-generating operators connect different eigenvalues of the Hamiltonian. The isotropic oscillator is a standard example of the relation between degeneracy and a symmetry larger than the rotation group \cite{jauch,fradkin}. For superintegrable systems, the corresponding integrals need not
generate a Lie algebra: their commutators can close polynomially, with the Hamiltonian and other commuting observables entering the coefficients. These algebraic methods, together with classification, separation of variables and special-function methods, are reviewed in \cite{review}.

Deformed-oscillator realizations can be used to construct representations of polynomial algebras. Generalized oscillator algebras encode ladder amplitudes in a structure function
\cite{dask91}. This construction has been used to determine spectra of two-dimensional quantum superintegrable systems from finite-dimensional representation conditions \cite{bonatsos}.
A systematic realization and Casimir construction for quadratic algebras are given in \cite{dask}. Deformed-oscillator realizations have also been developed for broader classes of
three-generator polynomial algebras \cite{isaac}. In this approach, a nonlinear commutator problem is converted into equations for functions of a number operator. Zeros and positivity of the structure function then select admissible unitary representations. When the energy occurs among the central parameters, these conditions can determine it.

In higher dimensions, the full symmetry algebra can contain several overlapping quadratic subalgebras. Such subalgebras and their Casimirs have been used in the three-dimensional Kepler--Coulomb problem \cite{tanoudis}. Chains of quadratic algebras and deformed oscillators have also been constructed for an $n$-dimensional system \cite{chain}. Centrality is relative to a selected subalgebra: an integral may be a parameter in one representation calculation
and a noncentral generator in the next.

For matrix Hamiltonians, finite spin identities supplement the differential-operator commutators, and several orbital channels can carry the same total angular momentum. In work dating from the
mid-1980s, dynamical and hidden superalgebras were used to study Pauli systems in monopole and dyon backgrounds. These constructions provided spectra and wave functions through representation theory and related accidental degeneracy to conserved operators
\cite{dhokervinet84,dhokervinet85cmp,dhokervinet85prl,
dhokervinet85npb,dhokervinet86}. These systems involve a different physical interaction and,
in their basic form, one Pauli spin rather than the two-spin singlet--triplet multiplicity considered here.

The study of three-dimensional superintegrable systems with one spin-$\tfrac12$ degree of freedom also includes systematic first- and second-order integral classifications \cite{spin2009,spin2012}. In the first-order setting, superintegrability was shown explicitly to lead to exact bound-state energies and wave functions for a nontrivial example \cite{spin2009}. More recently, the scalar and pseudo-scalar classifications for two spin-$\tfrac12$ particles have supplied larger families of spin-dependent Hamiltonians and their first-order integrals \cite{scalar,pseudo}. The subsequent vector classification also gives a representative scalar reduction with Coulomb- and oscillator-type solutions \cite{vector}. These results provide the starting point for asking how finite spin--orbital modules, their overlap and radial ladders combine in one spectral problem.

We study the representations and exact spectrum of an oscillator Hamiltonian in Case~1 of Ref.~\cite{pseudo}. The spin-dependent Hamiltonian family and its first-order scalar and pseudo-scalar integrals are inherited from the classification; the source paper also discusses their polynomial commutators. We use those operators to construct an explicit sequence of finite angular representations, determine their overlap, and combine them with a radial spectrum-generating representation on a specified Hilbert-space domain.

The angular calculation has two stages. A quadratic algebra acts on the three triplet orbital channels at fixed total angular momentum. An exchange algebra then joins the middle triplet to the singlet. The two modules share the middle triplet and span the four-channel space. The finite matrices, Casimir value and deformed-oscillator structure functions are derived explicitly. They determine angular--spin multiplicity and the exchange pairing, but contain no radial energy variable.

The radial spectrum is obtained separately. The oscillator factorization method \cite{infeld,arik} and the radial $\mathfrak{su}(1,1)$ construction \cite{radialsu11} are established tools. Here the centrifugal parameter is the total angular momentum of the coupled spinor problem, and the same radial representation is repeated on the angular multiplicity space. Its lowest weight is selected by the boundary condition and determines the infinite energy sequence. The radial raising and lowering operators change the energy; they are not additional integrals of motion.

The classified family permits an arbitrary radial scalar term. Adding the displayed scalar inverse-square contribution to the oscillator potential cancels the angular shift which would otherwise make the $j=0$ radial problem strongly attractive. This gives a common regular radial realization in all angular channels. We determine the energy-shell multiplicities in each particle-interchange sector and identify the angular symmetries that preserve those sectors.

The paper is organized as follows. Section~\ref{sec:model} introduces
our notation and Hamiltonian and derives the common radial problem.
Section~\ref{sec:angular} recalls the source integrals and obtains the
exchange and quadratic algebras. In Section~\ref{sec:finite} we construct
their finite-dimensional representations and determine how they act
together on the angular channels. Section~\ref{sec:radial} gives the
radial spectrum, normalized eigenfunctions and the common operator
domain. Section~\ref{sec:multiplicities} derives the energy-shell
multiplicities and determines the effect of particle interchange.
Finally, Section~\ref{sec:conclusion} summarizes the results.

\section{The Hamiltonian and angular decomposition}\label{sec:model}
\subsection{Hamiltonian and notation}
We consider the relative motion of two nonrelativistic spin-$\tfrac12$
particles. The reduced mass is set to one. The relative position is
$\x=(x_1,x_2,x_3)$, and the relative momentum and orbital angular momentum
operators are denoted by $\p$ and $\Lop$, respectively:
\begin{equation}
 \p=-\ii\hbar\nabla,\qquad \Lop=\x\times\p,\qquad r=|\x|,
 \qquad \hbar>0,\quad \omega>0.
\end{equation}
We write $\bm\sigma=(\sigma_x,\sigma_y,\sigma_z)$ for the Pauli vector, where
\[
 \sigma_x=\begin{pmatrix}0&1\\1&0\end{pmatrix},\qquad
 \sigma_y=\begin{pmatrix}0&-\ii\\\ii&0\end{pmatrix},\qquad
 \sigma_z=\begin{pmatrix}1&0\\0&-1\end{pmatrix}.
\]
The Pauli vectors for the two particles are
\[
 \sone=\bm\sigma\otimes\id_2,\qquad
 \stwo=\id_2\otimes\bm\sigma,
\]
where $\id_2$ is the $2\times2$ identity matrix. Their components are
$4\times4$ matrices acting on different spin factors and satisfy
$[\sigma_{1i},\sigma_{2j}]=0$. The corresponding spin operators are
$\hbar\sone/2$ and $\hbar\stwo/2$.
We use $\id$ for the identity operator and omit it from scalar terms when
its meaning is clear. The notations $(\bm a,\bm b)$ and
$\bm a\cdot\bm b$ both denote $\sum_{i=1}^{3}a_i b_i$, with the order of
operator factors preserved. We define
\begin{equation}\label{eq:spinnotation}
 K=\sone\cdot\stwo,\qquad
 \Sig=\sone+\stwo,\qquad \Del=\sone-\stwo,
 \qquad \Om=\sone\times\stwo.
\end{equation}
We consider the Hamiltonian
\begin{equation}\label{eq:Hmodel}
 H_\omega=\frac{\p^2}{2}+\frac{\omega^2r^2}{2}
 +\frac{\hbar}{2r^2}\Sig\cdot\Lop
 +\frac{\hbar^2}{4r^2}(K+3\id).
\end{equation}
In the notation of the source classification this means
\begin{equation}\label{eq:sourcepotentials}
 V_0=\frac{\omega^2r^2}{2}+\frac{3\hbar^2}{4r^2},\quad
 V_1=\frac{\hbar}{r^2},\quad V_2=\frac{\hbar^2}{4r^2},
 \quad V_3=V_4=V_5=0.
\end{equation}
This is an oscillator specialization of Case~1, in which the scalar potential is arbitrary. The scalar term $3\hbar^2/(4r^2)$ is part of the chosen potential $V_0$. For comparison, the representative radial reduction in Section~5 of Ref.~\cite{vector} has $V_2=0$. The spin--spin term and scalar counterterm in Eq.~\eqref{eq:sourcepotentials} select a different specialization; their common singlet--triplet radial reduction is derived below.

The Hilbert space is the full four-component relative-coordinate space
\begin{equation}
 \mathcal H=L^2(\mathbb R^3,d^3x)\otimes\mathbb C^2\otimes\mathbb C^2.
\end{equation}
No particle-interchange projection is imposed at this stage. Spin exchange acts on the two spin factors, whereas complete particle interchange also exchanges the particle positions. For identical spin-$\tfrac12$ fermions the physical relative-coordinate sector is antisymmetric under the complete interchange; the symmetric sector is retained below as a useful comparison. The consequences of such a projection are discussed in Section~\ref{sec:statistics}.

\subsection{Angular channels and their common radial Hamiltonian}\label{sec:channels}
Rotational invariance allows us to decompose the spinor space at fixed
$j,m$. The total angular momentum is
\begin{equation}
 \J=\Lop+\frac{\hbar}{2}\Sig,\qquad U=\Sig\cdot\Lop.
\end{equation}
The two spins combine to give total spin $s=0$ or $s=1$. Coupling either value to the orbital angular momentum gives $j=0,1,2,\ldots$. For $j\geq1$, the four multiplicity channels are
\begin{equation}\label{eq:channels}
 |s;jm\rangle=|(j,0)jm\rangle,\quad
 |t_-;jm\rangle=|(j-1,1)jm\rangle,\quad
 |t_0;jm\rangle=|(j,1)jm\rangle,\quad
 |t_+;jm\rangle=|(j+1,1)jm\rangle.
\end{equation}
The label $s$ on the first ket denotes the singlet. At $j=0$, only the singlet with $l=0$ and the triplet with $l=1$ exist. In particular there is no middle triplet at $j=0$.

The rotational eigenvalues are
\begin{equation}
 \J^2=\hbar^2j(j+1),\qquad J_3=\hbar m,
 \qquad j=0,1,\ldots,\quad m=-j,\ldots,j.
\end{equation}

The eigenvalues of $K$ are $-3$ in the singlet and $1$ in the triplet.
To obtain those of $U$, we use
\begin{equation}\label{eq:Jidentity}
 \J^2=\Lop^2+\hbar U+\frac{\hbar^2}{4}(6+2K).
\end{equation}
For a channel with orbital label $l$, substituting the eigenvalues
$\hbar^2j(j+1)$ and $\hbar^2l(l+1)$ of $\J^2$ and $\Lop^2$ gives
the following table:
\begin{center}
\begin{tabular}{llll}
\toprule
Channel & $l$ & $K$ & $U$ \\
\midrule
$s$ & $j$ & $-3$ & $0$ \\
$t_-$ & $j-1$ & $1$ & $2\hbar(j-1)$ \\
$t_0$ & $j$ & $1$ & $-2\hbar$ \\
$t_+$ & $j+1$ & $1$ & $-2\hbar(j+2)$ \\
\bottomrule
\end{tabular}
\end{center}
These eigenvalues will be used in both the angular representation
and the radial reduction.

Using $\p^2=-\hbar^2(\partial_r^2+2\partial_r/r)+\Lop^2/r^2$ in Eq.~\eqref{eq:Hmodel}, the angular numerator is
\begin{equation}
 \Lop^2+\hbar U+\frac{\hbar^2}{2}(K+3)=\J^2.
\end{equation}
Consequently
\begin{equation}\label{eq:Hradialfull}
 H_\omega=-\frac{\hbar^2}{2}\left(\partial_r^2+\frac2r\partial_r\right)
 +\frac{\J^2}{2r^2}+\frac{\omega^2r^2}{2}.
\end{equation}
For $u=rR$, the radial measure is $dr$, and the fixed-$j$ radial Hamiltonian becomes
\begin{equation}\label{eq:Hj}
 H_j=-\frac{\hbar^2}{2}\frac{d^2}{dr^2}
 +\frac{\hbar^2j(j+1)}{2r^2}+\frac{\omega^2r^2}{2}.
\end{equation}
The full Hamiltonian therefore acts as $H_j$ on the radial factor and as the identity on the multiplicity space in Eq.~\eqref{eq:channels}. This reduction holds at the operator level.

For $j=0$ we choose $u(0)=0$, the Friedrichs condition for the half-line oscillator. For $j\geq1$, the regular solution behaves as $r^{j+1}$ and the singular power $r^{-j}$ is not square-integrable at the origin. The oscillator potential is confining at infinity. These choices give the Friedrichs realization in every channel. The same radial condition is used for all channels at a fixed $j$.

Without the scalar counterterm, the common angular reduction of the unmodified Case-1 oscillator would replace $j(j+1)$ in Eq.~\eqref{eq:Hj} by $j(j+1)-3/2$. At $j=0$ this lies below the inverse-square semiboundedness threshold. The present Hamiltonian instead has coefficient zero there and the ordinary Dirichlet condition.

This common radial operator allows the angular integrals to mix channels without changing the radial boundary condition.

Let $\mathcal V_j$ denote the $(2j+1)$-dimensional irreducible
representation space of rotations with total angular momentum $j$,
with basis $|j,m\rangle$, $m=-j,\ldots,j$. Each spin--orbital channel
at fixed $j$ carries one copy of this representation. The full
angular decomposition is therefore
\begin{equation}\label{eq:methoddecomposition}
 \mathcal H\simeq\bigoplus_{j=0}^{\infty}
 \left[L^2(0,\infty;dr)\otimes\mathbb C^{\mu_j}
 \otimes\mathcal V_j\right],
 \qquad \mu_0=2,\quad\mu_j=4\ (j\geq1).
\end{equation}
Here $\mathbb C^{\mu_j}$ labels the distinct spin--orbital channels,
while $\mathcal V_j$ carries the magnetic states within each channel.
At fixed $j,m$, this reduces to
\begin{equation}\label{eq:architecturefactor}
 \mathcal H_{j,m}=L^2(0,\infty;dr)\otimes\mathcal M_j,
 \qquad
 \mathcal M_j=
 \begin{cases}
  \operatorname{span}\{s,t_-,t_0,t_+\},&j\geq1,\\
  \operatorname{span}\{s,t_+\},&j=0.
 \end{cases}
\end{equation}
The Hamiltonian acts as $H_j$ on the radial factor and as the identity on the finite factors. Angular scalar and pseudo-scalar integrals act on $\mathcal M_j$, while rotations act on $\mathcal V_j$.

\section{Angular integrals and polynomial algebras}\label{sec:angular}
The inherited scalar and pseudo-scalar integrals act within the angular
multiplicity spaces. Their commutators contain the exchange and quadratic
subalgebras needed to describe this action.

\subsection{Source integrals and working generators}\label{sec:complete-integrals}
We use the seven Case~1 pseudo-scalar integrals $Y_P^1,\ldots,Y_P^7$ from Ref.~\cite{pseudo}. This first-order classification includes both first-order and zeroth-order operators, as indicated in the table below. The scalar classification of Ref.~\cite{scalar} supplies the two first-order angular integrals $X$ and $Y$. Their commutator introduces the further conserved angular operator $Z$ used in the enlarged Case~1 algebra of Ref.~\cite{pseudo}:
\begin{equation}\label{eq:XYZsource}
 X=(\stwo,\Lop),\qquad Y=(\sone,\Lop),\qquad
 Z=(\sone\times\stwo,\Lop),\qquad [Y,X]=\ii\hbar Z.
\end{equation}
Together with $K=(\sone,\stwo)$ and the rotational integrals $\J^2,J_3$, all of these commute with the Hamiltonian:
\begin{equation}\label{eq:allcommuteH}
 [H_\omega,Y_P^a]=0\quad(a=1,\ldots,7),\qquad
 [H_\omega,X]=[H_\omega,Y]=[H_\omega,Z]=[H_\omega,K]=0,
\end{equation}
with $[H_\omega,\J^2]=[H_\omega,J_3]=0$.  The complete explicit differential formulas for $Y_P^1,\ldots,Y_P^7$ are collected in Appendix~\ref{app:source-integrals}; the four longer operators $Y_P^1,\ldots,Y_P^4$ are also displayed there as polynomial expressions in a smaller generating set.

For the representation calculations we introduce the combinations
\begin{equation}\label{eq:workinggenerators-early}
 U=X+Y,\qquad V=X-Y,\qquad Z_L=Z,\qquad
 B=2Y_P^6+Y_P^7,\qquad Q=Y_P^7,\qquad R=-Y_P^5,
\end{equation}
with $T=Y_P^2-Y_P^3$. The source labels and their working forms are
listed below:

\begin{center}
\begin{tabular}{@{}>{\raggedright\arraybackslash}p{0.14\linewidth}>{\raggedright\arraybackslash}p{0.34\linewidth}>{\raggedright\arraybackslash}p{0.10\linewidth}>{\raggedright\arraybackslash}p{0.30\linewidth}@{}}
\toprule
Source label & Explicit/working form & Order & Role below\\
\midrule
$Y_P^1$ & see Appendix~\ref{app:source-integrals} & first & polynomially dependent\\
$Y_P^2$ & see Appendix~\ref{app:source-integrals} & first & enters $T=Y_P^2-Y_P^3$\\
$Y_P^3$ & see Appendix~\ref{app:source-integrals} & first & enters $T=Y_P^2-Y_P^3$\\
$Y_P^4$ & see Appendix~\ref{app:source-integrals} & first & polynomially dependent\\
$Y_P^5$ & $(\Om,\x)/r$ & zero & $R=-Y_P^5$\\
$Y_P^6$ & $(\sone,\x)/r$ & zero & enters $B$ and $Q$\\
$Y_P^7$ & $(\stwo-\sone,\x)/r$ & zero & $Q=Y_P^7$\\
$X,Y$ & $(\stwo,\Lop),(\sone,\Lop)$ & first & $U=X+Y$, $V=X-Y$\\
$Z$ & $(\Om,\Lop)$ & first & $Z_L=Z$\\
\bottomrule
\end{tabular}
\end{center}
Under the spin-exchange involution $P=(K+\id)/2$, $B=2Y_P^6+Y_P^7=(\Sig,\x)/r$ is exchange-even, while $Q=Y_P^7=-(\Del,\x)/r$, $R=-Y_P^5$, $V=X-Y$ and $Z_L=Z$ form exchange-odd pairs.

The smaller source generating set
\begin{equation}\label{eq:sourceG}
 \mathcal G=\{K,Y_P^5,Y_P^6,Y_P^7,X,Y,Z\}
\end{equation}
generates an associative algebra with polynomial commutator closure.  The first relations already show why $K$ is not central in the full angular algebra:
\begin{align}\label{eq:sourcecomm1}
 &[K,Y_P^6]=2\ii Y_P^5,\quad [K,Y_P^5]=4\ii Y_P^7, \quad [K,Y_P^7]=-4\ii Y_P^5, \nonumber\\
 &[Y_P^6,Y_P^7]=0,\quad
 [Y_P^6,Y_P^5]=2\ii\left(K-\id+\frac{(Y_P^5)^2}{2}\right), \quad
 [Y_P^7,Y_P^5]=-4\ii\left(K-\id+\frac{(Y_P^5)^2}{2}\right).
\end{align}
The scalar angular generators satisfy
\begin{align}\label{eq:sourcecomm2}
 [K,Y]=2\ii Z,\quad [K,X]=-2\ii Z,\quad [K,Z]=4\ii(X-Y),\quad [Y,X]=\ii\hbar Z,
\end{align}
and their mixed commutators with the zero-order pseudo-scalars produce $Y_P^1,\ldots,Y_P^4$.  For example,
\begin{align}
 [Y,Y_P^6]=2\ii Y_P^3, \quad [X,Y_P^6]=-\ii\hbar Y_P^5,\quad
 [Y,Y_P^7]=\ii\hbar Y_P^5-2\ii Y_P^3,\quad
 [X,Y_P^7]=-2\ii Y_P^2+\ii\hbar Y_P^5.\label{eq:sourcemixedlast}
\end{align}
The polynomial reductions of $Y_P^1,\ldots,Y_P^4$ are given in Appendix~\ref{app:source-integrals}, and the remaining source commutators are collected in Appendix~\ref{app:source-commutators}.

\subsection{Exchange identities and the angular generators}\label{sec:exchange-algebras}
For the spin-exchange operator $P$, the Pauli relations give
\begin{equation}
 K^2+2K-3\id=0,\qquad P=\frac{K+\id}{2},\qquad P^2=\id.
\end{equation}
The triplet and singlet projectors are
\begin{equation}
 \Pi_t=\frac{K+3\id}{4},\qquad \Pi_s=\frac{\id-K}{4}.
\end{equation}
The exchange involution satisfies
\begin{equation}
 P\Del=\ii\Om,\qquad P\Om=-\ii\Del,\qquad [P,\Sig]=0.
\end{equation}
If a Hermitian integral $Q$ anticommutes with $P$, its Hermitian exchange partner is $R=-\ii P Q$. With $\mathcal Z=Q^2$, direct multiplication gives
\begin{align}
 [P,Q]=2\ii R,\quad [P,R]=-2\ii Q, \quad [Q,R]=2\ii\mathcal ZP,\nonumber \\
 Q^2=R^2=\mathcal Z, \quad [\mathcal Z,P]=[\mathcal Z,Q]=[\mathcal Z,R]=0.\label{eq:exchangesquare}
\end{align}
With $\bm n=\x/r$, the angular generators are
\begin{equation}\label{eq:generators}
 B=\Sig\cdot\bm n,\quad Q=-\Del\cdot\bm n,\quad R=-\ii P Q,
 \quad U=\Sig\cdot\Lop,\quad
 V=-\Del\cdot\Lop,\quad Z_L=\Om\cdot\Lop.
\end{equation}
All these operators are independent of the radial variable when expressed in angular coordinates. They are rotational scalars or pseudo-scalars and commute with $\J^2$. Equation~\eqref{eq:Hradialfull} therefore proves their formal conservation for $H_\omega$. Section~\ref{sec:domains} supplies the common spectral domain.

The position-type exchange pair satisfies
\begin{equation}\label{eq:positionids}
 BQ=BR=0,\qquad Q^2=R^2=4\id-B^2,
 \qquad B(B^2-4\id)=0.
\end{equation}
For $j\geq1$, the angular multiplicity space has a $B=0$ two-state exchange block and $B=\pm2$ one-dimensional blocks. At $j=0$, only the two-state $B=0$ block remains. These are kinematic Pauli identities. They do not depend on the oscillator frequency.

The orbital exchange pair instead satisfies
\begin{equation}\label{eq:Vexchange}
 Z_L=\ii P V,\qquad [P,V]=-2\ii Z_L,\qquad [P,Z_L]=2\ii V.
\end{equation}
The operator $Q$ connects the singlet to the outer triplet channels, whereas $V$ connects it to the middle triplet.

\subsection{The quadratic subalgebra}\label{sec:quadratic-algebra}
We next consider the exchange-even pair $U,B$. To identify its
commutator algebra and Casimir, we first recall the standard
three-generator quadratic form \cite{dask,isaac}:
\begin{align}\label{eq:methodquadratic}
 [A,B]&=C,\\
 [A,C]&=\alpha A^2+\gamma\{A,B\}+\delta A+\varepsilon B+\zeta,\\
 [B,C]&=aA^2-\gamma B^2-\alpha\{A,B\}+dA-\delta B+z.\label{eq:methodquadraticlast}
\end{align}
Here $\{A,B\}=AB+BA$, and the coefficients are central in the selected
algebra. On a sector with fixed central eigenvalues, they become
numerical parameters. For the present system, we set
\begin{equation}
 A=U,\qquad C=[A,B],\qquad \kappa=6+2K.
\end{equation}
Case~1 is a specialization of Case~2 of Ref.~\cite{pseudo}. Under that specialization, $B=(\Sig,\x/r)$ is the Case~2 generator $Y_P^8$ and $T=Y_P^2-Y_P^3$ is the Case~2 generator $Y_P^{10}$. The resulting subalgebra has the canonical quadratic form of Ref.~\cite{dask}:
\begin{align}\label{eq:quadratic}
 [A,B]&=C,\\
 [A,C]&=-2\hbar\{A,B\}+(4\J^2-\hbar^2\kappa)B,\label{eq:quadraticAC}\\
 [B,C]&=2\hbar B^2-4A-2\hbar\kappa.\label{eq:quadraticBC}
\end{align}
Comparison with Eqs.~\eqref{eq:methodquadratic}--\eqref{eq:methodquadraticlast}
gives $\alpha=\delta=\zeta=a=0$, $\gamma=-2\hbar$,
$\varepsilon=4\J^2-\hbar^2\kappa$, $d=-4$ and $z=-2\hbar\kappa$.
Both $K$ and $\J^2$ are central in this subalgebra.
The Casimir construction of Ref.~\cite{dask}, with these coefficients, gives
\begin{equation}\label{eq:CasQ}
 \mathcal K=C^2+2\hbar\{A,B^2\}
 +[\hbar^2(\kappa+4)-4\J^2]B^2-4A^2-4\hbar\kappa A.
\end{equation}
In the physical two-spin realization,
\begin{equation}\label{eq:CasQvalue}
 \mathcal K=4\kappa(2\hbar^2-\J^2).
\end{equation}
Here $U$ has differential order one and $B$ has order zero; quadratic refers to their polynomial commutator closure. The quadratic coefficients contain the angular invariant, but not $H_\omega$.

The generator $C$ satisfies $C=-2\ii T$, where the first-order pseudo-scalar $T$ is
\begin{equation}
 T=\frac1r(\Sig\cdot\x)(\x\cdot\p)-r\Sig\cdot\p
 -\frac{\ii\hbar}{r}\Sig\cdot\x.
\end{equation}
The apparent radial derivatives cancel in this combination. Thus the generators in Eqs.~\eqref{eq:quadratic}--\eqref{eq:quadraticBC} are angular on the reduced Hilbert space.

\section{Finite-dimensional representations}\label{sec:finite}
At fixed $j,m$, the angular algebras act on the finite multiplicity space
$\mathcal M_j$ in Eq.~\eqref{eq:architecturefactor}. We construct their
representations using the channel basis and a deformed-oscillator
realization.

\subsection{Deformed-oscillator realization}\label{sec:method}
A deformed oscillator encodes the ladder amplitudes in a structure
function $\Phi$. Its number operator $\mathsf N$ and mutually adjoint
pair $b,b^\dagger$ satisfy \cite{dask91,bonatsos,dask}
\begin{equation}\label{eq:methodoscillator}
 [\mathsf N,b^\dagger]=b^\dagger,\quad [\mathsf N,b]=-b,\quad
 b^\dagger b=\Phi(\mathsf N),\quad bb^\dagger=\Phi(\mathsf N+1).
\end{equation}
A standard realization is $A=A(\mathsf N)$ and $B=B_0(\mathsf N)+b^\dagger\rho(\mathsf N)+\rho(\mathsf N)b$, with real diagonal functions for a Hermitian representation. The commutator and Casimir identities determine these functions and the structure function $\Phi$. A finite unitary representation on $|\nu\rangle$, $\nu=0,\ldots,p$, requires
\begin{equation}\label{eq:methodendpoints}
 \Phi(0)=\Phi(p+1)=0,\qquad \Phi(\nu)>0\quad(1\leq\nu\leq p).
\end{equation}
The ladder amplitudes are $\sqrt{\Phi(\nu+1)}$ and $\sqrt{\Phi(\nu)}$; the upper zero means $b^\dagger|p\rangle=0$. These conditions determine angular multiplicities here, since the angular structure functions do not contain the energy. The number labels $N_a$ and $N_e$ used below are distinct from the radial number $n$.

The physical Pauli identities must also be satisfied. We choose the normalization $B=b+b^\dagger$ in the triplet calculation and compare its amplitudes with physical angular matrix elements. Other choices can redistribute diagonal factors between $b,b^\dagger$ and $\rho(\mathsf N)$ without changing the representation.

\subsection{Triplet representation of the quadratic subalgebra}
For $j\geq1$, the quadratic generators preserve the triplet sector.
In the ordered basis $t_-,t_0,t_+$, the generator $A$ is diagonal:
\begin{equation}\label{eq:Amatrix}
 A=\hbar\operatorname{diag}(2j-2,-2,-2j-4),
\end{equation}
while, in the common phase convention specified in Section~\ref{sec:complete-solution} and Appendix~\ref{app:physical-matrices},
\begin{equation}\label{eq:Bmatrix}
 B=\begin{pmatrix}
 0&2\sqrt{\dfrac{j+1}{2j+1}}&0\\
 2\sqrt{\dfrac{j+1}{2j+1}}&0&2\sqrt{\dfrac{j}{2j+1}}\\
 0&2\sqrt{\dfrac{j}{2j+1}}&0
 \end{pmatrix}.
\end{equation}
Together with $C=AB-BA$, these matrices satisfy all relations in Eqs.~\eqref{eq:quadratic}--\eqref{eq:quadraticBC}. They give
\begin{equation}
 \mathcal K=32\hbar^2[2-j(j+1)]\id_3.
\end{equation}
The singlet is a separate one-dimensional representation of this quadratic subalgebra, on which $A=B=C=0$ and $K=-3$.

The off-diagonal entries of $B$ connect neighbouring triplet channels.
We label these channels by $N_a=0,1,2$ and use the realization
\begin{equation}
 [N_a,b^\dagger]=b^\dagger,\quad [N_a,b]=-b,\quad
 b^\dagger b=\Phi_j(N_a),\quad bb^\dagger=\Phi_j(N_a+1),\qquad B=b+b^\dagger.
\end{equation}
The eigenvalue function of $A$ is
\begin{equation}\label{eq:AN}
 A(N_a)=\hbar\left[j(j+1)-\frac74-
 \left(N_a+j-\frac12\right)^2\right].
\end{equation}
At $N_a=0,1,2$ it reproduces Eq.~\eqref{eq:Amatrix}. Writing $x=N_a+j-\tfrac12$, the diagonal part of Eq.~\eqref{eq:quadraticBC} becomes
\begin{equation}\label{eq:recurrencePhi}
 (x+1)\Phi_j(N_a+1)-(x-1)\Phi_j(N_a)
 =-x^2+j(j+1)+\frac94.
\end{equation}
Solving with $\Phi_j(0)=0$ yields
\begin{equation}\label{eq:Phiangular}
 \Phi_j(N_a)=
 \frac{N_a(3-N_a)(2j+N_a-2)(2j+N_a+1)}
 {(2j+2N_a-3)(2j+2N_a-1)}.
\end{equation}
The values needed by the finite representation are
\begin{equation}\label{eq:Phiangularvalues}
 \Phi_j(0)=\Phi_j(3)=0,\qquad
 \Phi_j(1)=\frac{4(j+1)}{2j+1},\qquad
 \Phi_j(2)=\frac{4j}{2j+1}.
\end{equation}
Both interior values are positive for integer $j\geq1$, and their square roots are precisely the matrix entries in Eq.~\eqref{eq:Bmatrix}. The finite representation has $p=2$ and dimension three. The structure function is used only on this finite representation.

The rational form in Eq.~\eqref{eq:Phiangular} reflects the normalization $B=b+b^\dagger$. At the integer arguments used here, the denominators are nonzero and the finite amplitudes agree with the physical matrix elements.

The values $N_a=0,1,2$ label the three triplet orbital channels.

\subsection{Exchange representation, overlap and the full angular image}\label{sec:exchange-representation}
The quadratic algebra leaves the singlet separate. To include
singlet--triplet transitions, we use the exchange relations of
Section~\ref{sec:exchange-algebras}.

On a nonzero unitary block $\mathcal Z=z\id$, $z>0$,
the relations in Eq.~\eqref{eq:exchangesquare} become the Pauli
relations after dividing $Q,R$ by $\sqrt z$. The irreducible block is two-dimensional. The corresponding finite-oscillator structure function can be normalized as
\begin{equation}\label{eq:phiexchange}
 \Phi_{\mathrm{ex}}(N)=zN(2-N),\qquad N=0,1,\qquad \Phi_{\mathrm{ex}}(2)=0.
\end{equation}
At $z=0$, Hermiticity forces $Q=R=0$ on a unitary representation. This is a shortening of the exchange representation.

For the orbital exchange algebra, the combination
\begin{equation}\label{eq:Lambda}
 \Lambda=K+\frac2\hbar U
\end{equation}
commutes with $P,V,Z_L$. In fact
\begin{equation}
 [U,V]=2\ii\hbar Z_L,\qquad [U,Z_L]=-2\ii\hbar V,
\end{equation}
and the corresponding commutators with $K$ cancel in $\Lambda$. A useful polynomial identity is
\begin{equation}\label{eq:Vsquare}
 V^2=4\Lop^2-\frac{K+2}{3}(U^2+2\hbar U).
\end{equation}
The singlet and middle triplet both have $\Lambda=-3$, whereas the outer triplet values are
\begin{equation}
 \Lambda(t_-)=4j-3,\qquad \Lambda(t_+)=-4j-7.
\end{equation}
For $j\geq1$ the only nontrivial exchange action of $V$ is therefore on the singlet--middle-triplet pair. On that pair,
\begin{equation}\label{eq:Vactive}
 V^2=4\J^2=4\hbar^2j(j+1).
\end{equation}
The same construction applies to the orbital pair $(V,-Z_L)$.
With $z=4\hbar^2j(j+1)$, Eq.~\eqref{eq:phiexchange} gives the
$p=1$ structure function
\begin{equation}\label{eq:PhiV}
 \Phi_V(N_e)=4\hbar^2j(j+1)N_e(2-N_e),\qquad N_e=0,1,\quad\Phi_V(2)=0.
\end{equation}
The exchange number $N_e$ takes the values $0,1$, and $\Phi_V(2)=0$ is the upper termination condition.

At fixed $j,m$ with $j\geq1$, the triplet and orbital-exchange
representations act on
\[
\begin{aligned}
 \mathcal T_j
 &=\operatorname{span}\{t_-,t_0,t_+\}
 &&\text{(triplet quadratic algebra)},\\
 \mathcal E_j
 &=\operatorname{span}\{s,t_0\}
 &&\text{(orbital exchange algebra)}.
\end{aligned}
\]
The unital associative algebras generated by $U,B$ on $\mathcal T_j$
and by $P,V$ on $\mathcal E_j$ have matrix images
$M_3(\mathbb C)$ and $M_2(\mathbb C)$, respectively.
The shared channel is $t_0$, with $U=-2\hbar$ and $\Lambda=-3$.
Thus
\begin{equation}\label{eq:angular-overlap}
\begin{aligned}
 \mathcal T_j\cap\mathcal E_j
 &=\operatorname{span}\{t_0\},\\
 \mathcal T_j+\mathcal E_j
 &=\mathcal M_j
 =\operatorname{span}\{s,t_-,t_0,t_+\}.
\end{aligned}
\end{equation}
The dimension is therefore $3+2-1=4$. This matches the representation
labels through the common channel, as in the higher-dimensional
constructions \cite{tanoudis,chain}.

The operator $B$ connects $t_-$ to $t_0$ and $t_0$ to $t_+$,
while $V$ connects $s$ to $t_0$. Together with the channel projectors
obtained from $K,U$, these matrix elements generate $M_4(\mathbb C)$,
as shown in Proposition~\ref{prop:M4}.

\begin{proposition}\label{prop:M4}
For every integer $j\geq1$, the unital algebra generated by $K,U,B,V$ on the fixed-$j,m$ multiplicity space is $M_4(\mathbb C)$. For $j=0$, adjoining the position-type exchange generator $Q$ gives $M_2(\mathbb C)$ on the two existing channels.
\end{proposition}
\begin{proof}
For $j\geq1$, the singlet projector is a polynomial in $K$. Within the triplet, the three distinct eigenvalues of $U$ give three spectral projectors by polynomial interpolation. For example, with $u_-=2\hbar(j-1)$, $u_0=-2\hbar$, $u_+=-2\hbar(j+2)$,
\begin{equation}
 P_{t_-}=\Pi_t\frac{(U-u_0\id)(U-u_+\id)}{(u_--u_0)(u_--u_+)}.
\end{equation}
The denominators are nonzero at each fixed integer $j\geq1$. The other triplet projectors have the analogous form. The two nonzero off-diagonal entries of $B$ connect $t_-$ to $t_0$ and $t_0$ to $t_+$. The nonzero entry of $V$ connects $s$ to $t_0$. Multiplying on the left and right by the spectral projectors isolates each of these matrix units. Products and adjoints then produce every $E_{ab}$ because the four-channel graph is connected.

At $j=0$, only $s,t_+$ remain. In a suitable phase convention,
\begin{equation}\label{eq:j0matrices}
 P=\begin{pmatrix}-1&0\\0&1\end{pmatrix},\quad
 U=\begin{pmatrix}0&0\\0&-4\hbar\end{pmatrix},\quad
 Q=\begin{pmatrix}0&2\\2&0\end{pmatrix},\qquad B=V=0.
\end{equation}
The diagonal projector and $Q$ generate all two-by-two matrix units.
\end{proof}

At fixed $j,m$, the full matrix image implies that an angular operator commuting with all these generators is scalar on the multiplicity space.

At $j=0$ the $V$-based exchange oscillator vanishes because $\J^2=0$. The nontrivial two-state action is instead supplied by the position-type generator $Q$.

For $j=1$, in the order $s,t_-,t_0,t_+$,
\begin{equation}
 U=\hbar\operatorname{diag}(0,0,-2,-6),\qquad
 K=\operatorname{diag}(-3,1,1,1).
\end{equation}
The nonzero entries of $B$ are $B_{t_-,t_0}=2\sqrt{2/3}$ and $B_{t_0,t_+}=2/\sqrt3$, together with their adjoints. A compatible phase convention gives $V_{s,t_0}=2\sqrt2\hbar$. The triplet oscillator has
\begin{equation}
 \Phi_1(1)=\frac83,\qquad \Phi_1(2)=\frac43,
\end{equation}
and the overlapping exchange square is $V^2=8\hbar^2$ on its active pair. These four channels share every radial level at $j=1$. Including the three magnetic values gives twelve states for each fixed radial number in this sector.

\section{Exact spectrum and eigenfunctions}\label{sec:radial}
\subsection{The radial algebra, lowest weight and spectrum}
The common radial Hamiltonian $H_j$ in Eq.~\eqref{eq:Hj} admits the
standard $\mathfrak{su}(1,1)$ oscillator realization \cite{radialsu11}.
The defining relations are
\begin{equation}\label{eq:su11}
 [\mathsf K_0,\mathsf K_\pm]=\pm\mathsf K_\pm,
 \qquad [\mathsf K_-,\mathsf K_+]=2\mathsf K_0.
\end{equation}
The corresponding factorization methods are discussed in
Refs.~\cite{infeld,arik}. For the radial Hamiltonian $H_j$, a
differential realization is
\begin{align}\label{eq:radialgenerators}
 \mathsf K_0&=\frac{H_j}{2\hbar\omega},\\
 \mathsf K_\pm&=\frac{1}{4\hbar\omega}
 \left[\hbar^2\frac{d^2}{dr^2}+\omega^2r^2
 -\frac{\hbar^2j(j+1)}{r^2}
 \mp\hbar\omega\left(2r\frac{d}{dr}+1\right)\right].\label{eq:radialKpm}
\end{align}
These operators act on reduced radial functions in $L^2(0,\infty;dr)$. The adjoint of $r\partial_r$ on a suitable core is $-r\partial_r-1$, so $\mathsf K_+^\dagger=\mathsf K_-$ there. The Leibniz rule verifies Eq.~\eqref{eq:su11} for this realization.
The Casimir is
\begin{equation}\label{eq:radialCas}
 \mathsf C=\mathsf K_0(\mathsf K_0-1)-\mathsf K_+\mathsf K_-
 =\frac{j(j+1)}4-\frac3{16}.
\end{equation}
The angular generators commute with these radial operators: they act on the multiplicity factor and commute with $j(j+1)$. In the full reduced spinor space, $\hbar^2j(j+1)$ in Eq.~\eqref{eq:radialKpm} can be replaced by $\J^2$.

Since
\begin{equation}
 [H_j,\mathsf K_\pm]=\pm2\hbar\omega\mathsf K_\pm,
\end{equation}
the operators $\mathsf K_\pm$ raise and lower the radial energy and are not integrals of motion.

In a positive discrete-series representation, we write
\begin{equation}
 \mathsf K_0|n;j\rangle=(n+k_j)|n;j\rangle,\qquad
 \mathsf K_-|0;j\rangle=0,\qquad n=0,1,2,\ldots .
\end{equation}
The Casimir equation alone gives two algebraic roots for $k_j$. The physical root is selected by the radial boundary condition. Indeed, the regular lowest-weight function is
\begin{equation}\label{eq:u0}
 u_{0j}(r)=C_j r^{j+1}\exp\left(-\frac{\omega r^2}{2\hbar}\right).
\end{equation}
Substituting into Eqs.~\eqref{eq:radialgenerators} and \eqref{eq:radialKpm} gives
\begin{equation}
 \mathsf K_-u_{0j}=0,\qquad
 \mathsf K_0u_{0j}=\frac{2j+3}{4}u_{0j}.
\end{equation}
Therefore
\begin{equation}\label{eq:lowestweight}
 k_j=\frac{2j+3}{4}.
\end{equation}
The other root, $(1-2j)/4$, is nonpositive for integer $j\geq1$ and its corresponding radial power $r^{-j}$ is not square-integrable at the origin. At $j=0$, it is $1/4$ and corresponds to the different Neumann half-line oscillator, not the Dirichlet realization chosen here.

The structure function of the infinite radial representation is
\begin{equation}\label{eq:phiradial}
 \Phi_{\mathrm{rad},j}(n)=n(n+2k_j-1)=n\left(n+j+\frac12\right).
\end{equation}
It vanishes at the lowest state and is strictly positive for all integers $n\geq1$. Unlike Eqs.~\eqref{eq:Phiangular} and \eqref{eq:PhiV}, it has no finite upper truncation.

\begin{theorem}\label{thm:spectrum}
For the Friedrichs realization of Eq.~\eqref{eq:Hmodel} specified above, the complete bound-state energies are
\begin{equation}\label{eq:Enj}
 E_{n,j}=\hbar\omega\left(2n+j+\frac32\right),
 \qquad n,j=0,1,2,\ldots .
\end{equation}
The joint eigenspace with $H_\omega=E_{n,j}$ and $\J^2=\hbar^2j(j+1)$ has dimension $4(2j+1)$ for $j\geq1$ and two for $j=0$, including the magnetic values. The degeneracy of the energy alone requires the further sum in Section~\ref{sec:multiplicities}.
\end{theorem}
\begin{proof}
The energy follows from $H_j=2\hbar\omega\mathsf K_0$ and Eq.~\eqref{eq:lowestweight}. Each fixed-$j$ radial operator has the regular discrete-series basis, whose explicit complete Laguerre realization is given below. The angular decomposition of the full four-spinor space contains exactly the channels in Eq.~\eqref{eq:channels}, or the two actual channels at $j=0$. The radial operator is identical on these factors. Counting their dimensions and the $2j+1$ magnetic values gives the stated multiplicities. No additional radial channel is omitted by this decomposition.
\end{proof}

\subsection{Normalized eigenfunctions and ladder actions}
For a direct solution of the radial equation, we extract the regular
factor $r^{j+1}$ and the Gaussian decay \cite{arik,radialsu11}.
We set $a=\omega/\hbar$, $x=ar^2$, $\alpha_j=j+\tfrac12$ and
$u=r^{j+1}e^{-x/2}y(x)$. Substitution into Eq.~\eqref{eq:Hj} gives
\begin{equation}\label{eq:Laguerrecheck}
 xy''+(\alpha_j+1-x)y'+
 \left[\frac{E}{2\hbar\omega}-\frac{2j+3}{4}\right]y=0.
\end{equation}
For the solution regular at the origin, square integrability at infinity
requires the coefficient of $y$ to be a nonnegative integer $n$.
The polynomial solution is then $L_n^{(\alpha_j)}(x)$, and this
recovers Eq.~\eqref{eq:Enj}. The other local solution is excluded by
the boundary condition specified in Section~\ref{sec:channels}.

The normalization follows from the Laguerre orthogonality relation
and the radial change of variable:
\begin{align}
 \int_0^\infty x^{\alpha_j}e^{-x}
 L_n^{(\alpha_j)}(x)L_{n'}^{(\alpha_j)}(x)\,dx
 &=\frac{\Gamma(n+\alpha_j+1)}{n!}\delta_{nn'},\\
 r^{2j+2}\,dr&=\frac{x^{j+1/2}}{2a^{j+3/2}}\,dx.
\end{align}
Thus the normalized radial functions, with phases chosen to give
positive ladder amplitudes, are
\begin{equation}\label{eq:unormalized}
 u_{nj}(r)=(-1)^n
 \left[\frac{2a^{j+3/2}n!}{\Gamma(n+j+3/2)}\right]^{1/2}
 r^{j+1}e^{-ar^2/2}L_n^{(j+1/2)}(ar^2).
\end{equation}
They satisfy $\int_0^\infty u_{nj}^*u_{n'j}\,dr=\delta_{nn'}$.
The Laguerre polynomials are complete in the corresponding weighted
space, so these functions form a complete radial basis at fixed $j$.
For an orthonormal coupled angular spinor $\mathcal Y_{\eta}^{jm}$, the full state is
\begin{equation}\label{eq:fullstate}
 \Psi_{njm\eta}(\x)=\frac{u_{nj}(r)}{r}\mathcal Y_{\eta}^{jm}(\bm n).
\end{equation}
The angular spinors form a complete basis on the sphere with the two spin factors, while the Laguerre functions form the radial basis in each channel. Together these facts prove completeness of the displayed eigenfunctions for the specified realization. Since the energies tend to infinity with $n+j$ and each energy has finite multiplicity, this realization has no additional continuous spectrum.

The radial algebra acts by
\begin{align}\label{eq:radialactions}
 \mathsf K_+u_{nj}&=\sqrt{(n+1)(n+j+\tfrac32)}\,u_{n+1,j},\\
 \mathsf K_-u_{nj}&=\sqrt{n(n+j+\tfrac12)}\,u_{n-1,j},\label{eq:radialactionminus}\\
 \mathsf K_0u_{nj}&=\left(n+\frac{2j+3}{4}\right)u_{nj}.\label{eq:radialactionzero}
\end{align}
The angular operators leave $n,j$ unchanged and act by the finite matrices already displayed. The radial operators leave the angular channel and magnetic number unchanged. The two actions commute at fixed $j$.

\subsection{The complete spectral basis}\label{sec:complete-solution}
For the angular factor in Eq.~\eqref{eq:fullstate}, we use the
Condon--Shortley convention for the spherical harmonics and
Clebsch--Gordan coefficients. The spin states are
\begin{gather*}
 \chi_{00}=\frac{|\uparrow\downarrow\rangle-|\downarrow\uparrow\rangle}{\sqrt2},
 \qquad \chi_{11}=|\uparrow\uparrow\rangle,\\
 \chi_{10}=\frac{|\uparrow\downarrow\rangle+|\downarrow\uparrow\rangle}{\sqrt2},
 \qquad \chi_{1,-1}=|\downarrow\downarrow\rangle.
\end{gather*}
The common channel phases used in the angular matrices are fixed by
\begin{align}
 \mathcal Y_s^{jm}(\bm n)&=Y_{jm}(\bm n)\chi_{00},\label{eq:coupled-singlet}\\
 \mathcal Y_{t_\nu}^{jm}(\bm n)&=
 \varepsilon_\nu\sum_{m_l,m_s}\langle l_\nu m_l,1m_s|jm\rangle
 Y_{l_\nu m_l}(\bm n)\chi_{1m_s},\label{eq:coupled-triplet}
\end{align}
where $l_-=j-1$, $l_0=j$, $l_+=j+1$, and
$\varepsilon_-=\varepsilon_+=1$, $\varepsilon_0=-1$.
Appendix~\ref{app:physical-matrices} gives the corresponding physical matrix elements.

For $j\geq1$ all four angular labels $\eta=s,t_-,t_0,t_+$ occur. For $j=0$, angular momentum coupling leaves only $\eta=s,t_+$. With $a=\omega/\hbar$, a complete orthonormal bound-state basis of the specified realization is given by
\begin{equation}\label{eq:completewavefunction}
 \Psi_{njm\eta}(\x)=(-1)^n
 \left[\frac{2a^{j+3/2}n!}{\Gamma(n+j+3/2)}\right]^{1/2}
 r^j e^{-ar^2/2}L_n^{(j+1/2)}(ar^2)
 \mathcal Y_{\eta}^{jm}(\bm n).
\end{equation}
The allowed quantum numbers and energies are
\begin{equation}\label{eq:completequantumnumbers}
 n=0,1,2,\ldots,\qquad j=0,1,2,\ldots,\qquad
 m=-j,-j+1,\ldots,j,
 \qquad
 E_{n,j}=\hbar\omega\left(2n+j+\frac32\right),
\end{equation}
with the channel restrictions just stated. The state is a $K=-3$ singlet for $\eta=s$ and a $K=1$ triplet for $\eta=t_\nu$. The operator $U=\Sig\cdot\Lop$ has eigenvalues
\begin{equation}
 U(s)=0,\qquad U(t_-)=2\hbar(j-1),\qquad
 U(t_0)=-2\hbar,\qquad U(t_+)=-2\hbar(j+2).
\end{equation}
The three number states of the quadratic oscillator correspond to $t_-,t_0,t_+$. The $p=1$ exchange oscillator acts on $s,t_0$, whereas the position-type exchange operator $Q$ links the singlet to a definite combination of the outer triplets and supplies the nontrivial $j=0$ exchange action.

For fixed $(n,j)$, including all magnetic values, the multiplicity is
\begin{equation}\label{eq:fixednjmultiplicity}
 d_{n,j}=d_j=
 \begin{cases}
 2,&j=0,\\
 4(2j+1),&j\geq1.
 \end{cases}
\end{equation}
This is the multiplicity at fixed $n,j$. The degeneracy of the energy
alone also includes distinct pairs $(n,j)$ with the same value of
$2n+j$, as discussed in Section~\ref{sec:multiplicities}.

\subsection{Operator domains and the combined representation}\label{sec:domains}
The angular and radial operators must act on a common invariant domain
compatible with the chosen boundary condition. The finite linear span
of the states in Eq.~\eqref{eq:fullstate} provides a common invariant core. On this core, all angular integrals preserve $n,j$ and satisfy their exact finite matrix relations. The radial actions in Eqs.~\eqref{eq:radialactions}--\eqref{eq:radialactionzero} take a finite combination to another finite combination. Both actions preserve the regular radial branch.

The Hamiltonian is the orthogonal direct sum of the Friedrichs radial operators $H_j$, tensored with the existing finite angular multiplicity spaces and magnetic factors. At each fixed $j,m$, $U,V$ are Hermitian finite matrices, while $P,B,Q,R$ are bounded finite matrices. Their natural direct-sum closures give the corresponding angular operators. Because they act inside the energy eigenspaces, they commute with the spectral projections of $H_\omega$ on their natural domains.

The radial ladder pair has the adjoint relation and positive discrete-series action described above. At fixed $j,m$, the combined representation can be written schematically as
\begin{equation}\label{eq:fullrepresentedstructure}
 D^+(k_j)\otimes\mathbb C^4\quad(j\geq1),\qquad
 D^+(\tfrac34)\otimes\mathbb C^2\quad(j=0).
\end{equation}
On the common invariant core, the first factor carries the radial spectrum-generating algebra, and the second carries the represented angular algebra. The two actions commute.

The table summarizes these representations on the unprojected
relative-coordinate Hilbert space.
\begin{center}
\begin{tabular}{@{}>{\raggedright\arraybackslash}p{0.19\linewidth}>{\raggedright\arraybackslash}p{0.23\linewidth}>{\raggedright\arraybackslash}p{0.22\linewidth}>{\raggedright\arraybackslash}p{0.26\linewidth}@{}}
\toprule
Structure & Generators & Central labels & Information obtained\\
\midrule
Triplet quadratic algebra & $A=U$, $B$, $C=[A,B]$ & $K,\J^2$ & three triplet orbital channels; $p=2$ finite oscillator\\
Middle exchange algebra & $P,V,Z_L$ & $\Lambda=K+2U/\hbar$, $\J^2$ & singlet--middle-triplet pair; $p=1$ finite oscillator\\
Position exchange algebra & $P,Q,R$ & $B$ (with $Q^2=R^2=4\id-B^2$), $\J^2$ & complementary singlet--outer exchange and nontrivial $j=0$ block\\
Full angular image & $K,U,B,V$ for $j\geq1$; $K,U,Q$ at $j=0$ & fixed $j,m$ & $M_4(\mathbb C)$ for $j\geq1$; $M_2(\mathbb C)$ for $j=0$\\
Radial dynamical algebra & $\mathsf K_0,\mathsf K_\pm$ & fixed $j$ & positive discrete series $D^+(k_j)$ and the energy ladder\\
\bottomrule
\end{tabular}
\end{center}

\section{Level multiplicities and particle statistics}\label{sec:multiplicities}
\subsection{Joint angular multiplicities and complete energy shells}
The multiplicity $d_j$ counts states at fixed $n,j$. To obtain the
degeneracy of an energy level, we must add all such spaces with the same
value of $2n+j$. We therefore introduce the shell number
\begin{equation}
 N=2n+j,\qquad E_N=\hbar\omega(N+\tfrac32).
\end{equation}
At fixed $N$, the allowed total angular momenta are $j=N,N-2,\ldots$. If $d_j$ denotes the angular-spin multiplicity including $m$, then
\begin{equation}
 d_0=2,\qquad d_j=4(2j+1)\quad(j\geq1),\qquad
 D_N=\sum_{j=N,N-2,\ldots}d_j.
\end{equation}
We distinguish the even and odd shells by writing $N=2q$ or $N=2q+1$. The $j=0$ term contributes two states in the even case and is absent in the odd case:
\begin{align}\label{eq:evenoddcounts}
 D_{2q}&=2+4\sum_{a=1}^{q}(4a+1)=8q^2+12q+2,\\
 D_{2q+1}&=4\sum_{a=0}^{q}(4a+3)=8q^2+20q+12.
\end{align}
Equivalently,
\begin{equation}\label{eq:DN}
 D_N=2(N+1)(N+2)-2\,\mathbf1_{\{N\ \mathrm{even}\}}.
\end{equation}
For example, at $N=2$ the $j=2$ sector contributes twenty states and the $j=0$ sector contributes two. The two-state correction in every even shell is the consequence of the actual $j=0$ angular coupling.

Equation~\eqref{eq:DN} counts the full unprojected relative-coordinate spectrum. No center-of-mass oscillator or further internal degree of freedom has been included.

The multiplicity generating function is
\begin{equation}\label{eq:GF}
 \sum_{N=0}^\infty D_Nz^N=\frac4{(1-z)^3}-\frac2{1-z^2},\qquad |z|<1.
\end{equation}
Indeed,
\begin{equation}
 \sum_{j=0}^{\infty}d_jz^j
 =4\frac{1+z}{(1-z)^2}-2,
\end{equation}
and multiplication by $1/(1-z^2)$ accounts for the sum over the radial quantum number. The subtraction is exactly the missing two channels at $j=0$, not a correction at a single energy. Thus the radial and angular counts agree independently. A further check separates the spin sectors. The singlet is an ordinary scalar oscillator and contributes $(N+1)(N+2)/2$ states. The triplet contributes $3(N+1)(N+2)/2-2\mathbf1_{\{N\ \mathrm{even}\}}$, since its $j=0$ multiplicity is one rather than three. Their sum again gives Eq.~\eqref{eq:DN}.

The scalar three-dimensional oscillator has the usual $SU(3)$ degeneracy interpretation \cite{jauch,fradkin}. Here $j$ is total angular momentum, and the reduced channel count at $j=0$ prevents an interpretation as four independent scalar oscillators. The constructed angular generators act at fixed $j$, while the radial ladders change the energy.

\subsection{Particle statistics, multiplicities and surviving angular symmetries}\label{sec:statistics}
For identical spin-$\tfrac12$ fermions without additional flavor or
isospin degrees of freedom, only states antisymmetric under complete
particle interchange are allowed. The symmetric sector is included
for comparison.

The involution $P$ exchanges the spin factors only. Let $\mathscr R$
denote inversion of the relative coordinate and let
$\mathscr P=\mathscr R P$ denote complete particle interchange.
The spatial parity of a channel is $(-1)^l$, while spin exchange
has eigenvalue $-1$ on the singlet and $+1$ on the triplet.
Combining these signs gives, at fixed $j\geq1$ in the order
$s,t_-,t_0,t_+$,
\begin{equation}\label{eq:statsgrading}
 \mathscr P\big|_{j,m}=(-1)^j\operatorname{diag}(-1,-1,+1,-1).
\end{equation}
The space $\mathcal W_3=\operatorname{span}\{s,t_-,t_+\}$ has one
interchange eigenvalue, and $t_0$ has the opposite one.
For even $j\geq2$, the three channels in $\mathcal W_3$ are
antisymmetric and $t_0$ is symmetric; for odd $j$ their roles are
reversed. At $j=0$, both existing channels $s,t_+$ are antisymmetric.

We obtain the statistics-resolved multiplicities by retaining the
appropriate channels, including the magnetic factor $2j+1$, and
summing over $j=N,N-2,\ldots$ as in the preceding subsection.
The antisymmetric multiplicity is
\begin{equation}
 D_N^{(-)}=
 \begin{cases}
 \dfrac32(N+1)(N+2)-1,&N\text{ even}, \smallskip\\
 \dfrac12(N+1)(N+2),&N\text{ odd},
 \end{cases}
\end{equation}
and the symmetric multiplicity is
\begin{equation}\label{eq:symmetriccounts}
 D_N^{(+)}=
 \begin{cases}
 \dfrac12(N+1)(N+2)-1,&N\text{ even},\smallskip \\
 \dfrac32(N+1)(N+2),&N\text{ odd}.
 \end{cases}
\end{equation}
Their sum is Eq.~\eqref{eq:DN}. The first shells are listed in
Table~\ref{tab:shellcounts}; the pairs in its second column are $(n,j)$.
These counts impose particle statistics on the unprojected spectrum
without changing the energy formula or the radial functions.

\begin{table}[htbp]
\centering
\begin{tabular}{@{}c l r r r@{}}
\toprule
$N$ & Allowed $(n,j)$ & $D_N$ & $D_N^{(-)}$ & $D_N^{(+)}$\\
\midrule
0 & $(0,0)$ & 2 & 2 & 0\\
1 & $(0,1)$ & 12 & 3 & 9\\
2 & $(0,2),(1,0)$ & 22 & 17 & 5\\
3 & $(0,3),(1,1)$ & 40 & 10 & 30\\
4 & $(0,4),(1,2),(2,0)$ & 58 & 44 & 14\\
5 & $(0,5),(1,3),(2,1)$ & 84 & 21 & 63\\
\bottomrule
\end{tabular}
\caption{Full energy-shell degeneracy and its antisymmetric/symmetric parts. Every row has $E_N=\hbar\omega(N+3/2)$. The magnetic factor is already included.}\label{tab:shellcounts}
\end{table}

Conservation of energy alone does not ensure preservation of particle
statistics. To preserve the two interchange sectors, an angular integral
must also commute with $\mathscr P$.
The position-type exchange operator $Q$ is odd under both spatial parity and spin exchange, so it commutes with their product. The angular scalars $P,U$ also commute with $\mathscr P$. In contrast, the triplet pseudo-scalar $B$ and the exchange-odd scalar $V$ reverse $\mathscr P$:
\begin{equation}
 [\mathscr P,P]=[\mathscr P,U]=[\mathscr P,Q]=0,
 \qquad \{\mathscr P,B\}=\{\mathscr P,V\}=0.
\end{equation}
Thus the original nested construction with $B,V$ describes the unprojected four-spinor space. Its individual generators are not all observables in one fixed-statistics sector.

\begin{proposition}\label{prop:statsalgebra}
For $j\geq1$, the fixed-$j,m$ associative algebra generated by $P,U,Q$ is
\begin{equation}\label{eq:statsM3}
 \operatorname{Alg}(P,U,Q)\big|_{j,m}
 =M_3(\mathbb C)\oplus\mathbb C,
\end{equation}
with $M_3$ acting on $\mathcal W_3$ and the scalar block acting on $t_0$. It is exactly the part of the full $M_4$ angular image that commutes with $\mathscr P$. At $j=0$ its image is $M_2(\mathbb C)$, entirely in the antisymmetric sector.
\end{proposition}
\begin{proof}
Equation~\eqref{eq:statsgrading} gives the upper bound $M_3\oplus\mathbb C$. To prove that the selected generators reach it, use $P$ to isolate the singlet and triplet. The three eigenvalues $u_-,u_0,u_+$ of $U$ are distinct inside the triplet for every $j\geq1$, so the individual triplet projectors are
\begin{equation}
 P_a=\frac{1+P}{2}\prod_{b\ne a}\frac{U-u_b}{u_a-u_b},
 \qquad a,b\in\{-,0,+\}.
\end{equation}
The singlet projector is $\Pi_s=(1-P)/2$. The nonzero entries of $Q$ connect $s$ to both $t_-$ and $t_+$. Sandwiching $Q$ between the projectors yields the corresponding matrix units; their products give the remaining units on $\mathcal W_3$. The projector onto $t_0$ is the independent scalar block. Even at $j=1$, where the singlet and $t_-$ have equal $U$ eigenvalues, $P$ separates them. At $j=0$, $P$ and the nonzero two-state $Q$ give $M_2$ directly.
\end{proof}

This $M_3$ acts on the singlet and the two outer triplets.

The three-state block can also be expressed as a finite oscillator
using the matrix units just constructed. On $\mathcal W_3$ we use the
ordered basis $t_-,s,t_+$ and label its states $0,1,2$. Let $E_{ab}$
denote the normalized matrix units in this basis. The operators
\begin{equation}
 N_{\rm st}=E_{11}+2E_{22},\qquad
 b_{\rm st}^\dagger=\sqrt2(E_{10}+E_{21}),\qquad
 b_{\rm st}=\sqrt2(E_{01}+E_{12})
\end{equation}
provide a finite oscillator with $b_{\rm st}^\dagger b_{\rm st}=N_{\rm st}(3-N_{\rm st})$. These operators are normalized within the fixed-$j$ matrix representation and differ from the original quadratic pair $(U,B)$.

The surviving representations are therefore
\begin{center}
\begin{tabular}{lll}
\toprule
Total angular momentum & Antisymmetric sector & Symmetric sector\\
\midrule
$j=0$ & two-state $M_2$ & absent\\
$j\geq2$ even & three-state $M_3$ & one-dimensional middle triplet\\
$j\geq1$ odd & one-dimensional middle triplet & three-state $M_3$\\
\bottomrule
\end{tabular}
\end{center}
Each entry is repeated for the $2j+1$ magnetic values and the allowed radial levels. The radial $\mathfrak{su}(1,1)$ action commutes with $\mathscr P$ and survives in every nonempty sector. Summing this table at $N=2n+j$ reproduces the symmetric and antisymmetric shell counts above. Additional isospin or other internal variables would change the statistics restriction and are not part of this model.

\section{Discussion and conclusions}\label{sec:conclusion}
We have constructed representations of the inherited angular integrals
for the oscillator Hamiltonian in Eq.~\eqref{eq:Hmodel}. The quadratic
algebra acts on the three triplet channels through a finite oscillator
with $p=2$. The orbital exchange algebra couples the singlet to the
middle triplet through a $p=1$ representation. Their common channel
allows the two actions to combine, and the resulting matrix image is
$M_4(\mathbb C)$ at fixed $j,m$ for $j\geq1$. At $j=0$, the
position-type exchange generator instead gives $M_2(\mathbb C)$
on the two existing channels.

The common radial Hamiltonian is solved independently through its
$\mathfrak{su}(1,1)$ representation. The Friedrichs boundary condition
selects the lowest weight and hence the energy spectrum. Combining
the radial functions with the angular spinors gives the normalized
spectral basis and the ladder actions of Section~\ref{sec:radial}.
Summing the angular channels then gives the complete shell
multiplicities. Complete particle interchange restricts these states
and leaves the statistics-preserving part of the generic angular image
$M_3(\mathbb C)\oplus\mathbb C$. The radial algebra remains valid in
each nonempty statistics sector.

Relative to the scalar-potential applications of Refs.~\cite{bonatsos,dask,tanoudis,chain}, the characteristic feature is the finite matrix multiplicity remaining after total-angular-momentum decomposition. Spin exchange supplies a simple way to organize that multiplicity, while the central parameters change between the selected subalgebras. Relative to the first-order spin classifications \cite{scalar,pseudo,vector}, the emphasis is on realizing the inherited integrals on a complete spectral basis and determining which parts survive a statistics restriction. The radial $\mathfrak{su}(1,1)$ method itself is standard \cite{radialsu11}; its combination with these angular modules yields the stated energy and multiplicity result for the displayed Hamiltonian.

For a different scalar potential in the same family the angular representation may persist while radial solvability does not.

\appendix
\section{Complete Case 1 source integrals}\label{app:source-integrals}
The seven Case~1 pseudo-scalar integrals used in
Section~\ref{sec:complete-integrals} are, in the normalization of
Ref.~\cite{pseudo} and with $r=|\x|$ and $\p=-\ii\hbar\nabla$,
\begin{align}
 Y_P^1&=-\frac1r(\sone,\x)(\stwo,\Lop)
 +\frac{\ii\hbar}{2r}(\x,\sone\times\stwo),\label{eq:YP1app}\\
 Y_P^2&=\frac1r(\stwo,\x)(\x,\p)-r(\stwo,\p)
 -\frac{\ii\hbar}{r}(\stwo,\x),\label{eq:YP2app}\\
 Y_P^3&=r(\sone,\p)-\frac1r(\sone,\x)(\x,\p)
 +\frac{\ii\hbar}{r}(\sone,\x),\label{eq:YP3app}\\
 Y_P^4&=-\frac1r\Big[(\stwo,\x)(\sone,\Lop)
 -(\sone,\x)(\stwo,\Lop)
 +\ii\hbar(\x,\sone\times\stwo)\Big],\label{eq:YP4app}\\
 Y_P^5&=\frac{(\x,\sone\times\stwo)}{r},\qquad
 Y_P^6=\frac{(\sone,\x)}{r},\qquad
 Y_P^7=\frac{(\stwo-\sone,\x)}{r}.\label{eq:YP567app}
\end{align}
For the oscillator specialization used in this paper, all seven remain integrals because the scalar potential $V_0(r)$ is arbitrary in Case~1.  The first four are not required as independent generators of the finitely generated source algebra.  With $X=(\stwo,\Lop)$ and $Y=(\sone,\Lop)$, Ref.~\cite{pseudo} gives
\begin{align}
 Y_P^1&=-Y_P^6X+\frac{\ii\hbar}{2}Y_P^5,\label{eq:YP1poly}\\
 Y_P^2&=-\ii(Y_P^6+Y_P^7)(X+\hbar),\label{eq:YP2poly}\\
 Y_P^3&=\ii Y_P^6(Y+\hbar),\label{eq:YP3poly}\\
 Y_P^4&=Y_P^6X-(Y_P^6+Y_P^7)Y-\ii\hbar Y_P^5.\label{eq:YP4poly}
\end{align}
The identity
\begin{equation}
 \frac{1}{r^2}(\sone,\x)(\stwo,\x)=Y_P^6(Y_P^6+Y_P^7)
\end{equation}
is another useful product reduction.  These relations explain why the main text can work with the smaller set $\mathcal G$ in Eq.~\eqref{eq:sourceG} without losing any of the seven pseudo-scalars.

\section{Source commutators and the selected subalgebras}\label{app:source-commutators}
In addition to Eq. \eqref{eq:sourcemixedlast},
the source integrals satisfy
\begin{align}
 [Y,Y_P^5]&=2\ii Y_P^1-2\ii\hbar(Y_P^6+Y_P^7),\label{eq:sourceapp1}\\
 [X,Y_P^5]&=-2\ii(Y_P^1+Y_P^4)+2\ii\hbar Y_P^6,\label{eq:sourceapp2}\\
 [Y,Z]&=2\ii K\Lop^2-\ii(YX+XY)-2\ii\hbar X,\label{eq:sourceapp3}\\
 [X,Z]&=-2\ii K\Lop^2+\ii(YX+XY)+2\ii\hbar Y,\label{eq:sourceapp4}
\end{align}
where
\begin{equation}
 \Lop^2=\J^2-\hbar(X+Y)-\frac{\hbar^2}{4}(6+2K).
\end{equation}
Together with Eqs.~\eqref{eq:YP1poly}--\eqref{eq:YP4poly}, these relations close polynomially on $\mathcal G$.  Every member of $\mathcal G$ commutes with $H_\omega$ and with the total rotational Casimir $\J^2$.  The role of $K$ depends on the selected subalgebra: it is central in the triplet quadratic algebra generated by $U=X+Y$, $B=(\Sig,\x)/r$ and $T=Y_P^2-Y_P^3$, but it is noncentral in the exchange sectors.  Conversely,
\begin{equation}
 \Lambda=K+\frac{2U}{\hbar}
\end{equation}
is central in the $(P,V,Z_L)$ exchange algebra.

In the physical angular representation, the nontrivial mixed commutators among $U,V,Z_L$ and $B,Q,R$ in the working basis are
\begin{align}
 [U,B]&=-BU-QV-2\hbar B, \nonumber &
 [U,Q]&=-BV-QU-2\hbar Q-2\ii\hbar R,\nonumber  \\
 [U,R]&=\ii BV-\ii QU+2\ii\hbar Q-2\hbar R, &
 [V,B]&=-BV-QU-2\hbar Q+2\ii\hbar R, \nonumber \\
 [V,Q]&=-BU-QV-2\hbar B, &
 [V,R]&=-\ii BU+\ii QV-2\ii\hbar B, \nonumber \\
 [Z_L,B]&=-\ii BV+\ii QU+2\ii\hbar Q+2\hbar R, &
 [Z_L,Q]&=-\ii BU+\ii QV-2\ii\hbar B,\nonumber  \\
 [Z_L,R]&=BU+QV+2\hbar B.\label{eq:workingmixedledger}
\end{align}
Together with the internal quadratic and exchange relations, these give the commutators in the working basis.

The generators and their representation roles are summarized in the
table at the end of Section~\ref{sec:domains}.

\section{Full angular matrices and product checks}\label{app:physical-matrices}
We compute the angular matrices in the common phase convention of
Eqs.~\eqref{eq:coupled-singlet}--\eqref{eq:coupled-triplet}, using a
vector-harmonic basis.

We use the four-channel order $s,t_-,t_0,t_+$ and write $q_j=\hbar^2j(j+1)$. We set
\begin{equation}
 b_-=2\sqrt{\frac{j+1}{2j+1}},\qquad
 b_+=2\sqrt{\frac{j}{2j+1}},\qquad v_j=2\hbar\sqrt{j(j+1)}.
\end{equation}
For the coupled spinors in Eqs.~\eqref{eq:coupled-singlet}--\eqref{eq:coupled-triplet}, let $\chi_s=\chi_{00}$ and $\chi_i=\Delta_i\chi_s/2$. This is an orthonormal Cartesian triplet basis, with $\Sigma_i\chi_j=2\ii\epsilon_{ijk}\chi_k$. We identify $\sum_i v_i\chi_i$ with $\bm v$.

For $j\geq1$, we set $g=j(j+1)$, $c_j=\sqrt{j/(2j+1)}$ and $d_j=\sqrt{(j+1)/(2j+1)}$. With $\nabla_S$ denoting the surface gradient, we define
\[
 \bm e_r=\bm nY_{jm},\qquad
 \bm e_g=\frac{\nabla_S Y_{jm}}{\sqrt g},\qquad
 \bm e_c=\frac{\ii\bm n\times\nabla_S Y_{jm}}{\sqrt g}.
\]
These vector harmonics are orthonormal by $-\Delta_SY_{jm}=gY_{jm}$ and integration by parts on the sphere. In the stated phase convention the triplet channels are
\[
 t_-=c_j\bm e_r+d_j\bm e_g,\qquad
 t_0=\bm e_c,\qquad
 t_+=-d_j\bm e_r+c_j\bm e_g.
\]
A componentwise application of $\Lop^2$ gives orbital angular momenta $j-1,j,j+1$, respectively. The physical actions are
\[
 B\bm v=2\ii\bm n\times\bm v,\qquad
 Q(Y_{jm}\chi_s)=-2\bm e_r,\qquad
 V(Y_{jm}\chi_s)=2\hbar\sqrt g\,\bm e_c.
\]
Since $B$ annihilates the singlet and $Q,V$ are spin-exchange odd, these actions and Hermiticity determine the displayed matrices. At $j=0$, only $Y_{00}\chi_s$ and $t_+=-\bm nY_{00}$ remain, giving $Q_{s,t_+}=2$ and $B=V=0$.

For $j\geq1$, this gives the matrices
\begin{align*}
 U&=\operatorname{diag}(0,2\hbar(j-1),-2\hbar,-2\hbar(j+2)),\quad 
 P=\operatorname{diag}(-1,1,1,1),\\
 B&=\begin{pmatrix}0&0&0&0\\0&0&b_-&0\\0&b_-&0&b_+\\0&0&b_+&0\end{pmatrix},\quad
 Q=\begin{pmatrix}0&-b_+&0&b_-\\-b_+&0&0&0\\0&0&0&0\\b_-&0&0&0\end{pmatrix},\quad
 V=\begin{pmatrix}0&0&v_j&0\\0&0&0&0\\v_j&0&0&0\\0&0&0&0\end{pmatrix}.
\end{align*}
The remaining matrices are $R=-\ii P Q$, $Z_L=\ii P V$ and $C=UB-BU$. A rephasing of any channel conjugates all matrices simultaneously; changing a sign in only one of them would not in general preserve the mixed identities.

These matrices satisfy
\begin{align*}
 &BQ=0,\quad B^2+Q^2=4\id, \quad R^2=Q^2, \quad
 [U,V]=2\ii\hbar Z_L,\\ & [U,Z_L]=-2\ii\hbar V,\quad
 [U,B]=-BU-QV-2\hbar B.
\end{align*}
They display both the two exchange directions and their interaction with the triplet quadratic algebra. Equation~\eqref{eq:Vsquare} follows by using $\Lop^2=q_j\id-\hbar U-\hbar^2(6+2K)/4$ in these matrices. The Casimir calculation is similarly a direct multiplication and gives Eq.~\eqref{eq:CasQvalue} separately on the singlet and triplet blocks.

The proof of Proposition~\ref{prop:M4} gives an explicit construction of the matrix units. If $P_a$ is a channel projector and $G$ is a generator with nonzero entry $g_{ab}$, then
\begin{equation}
 E_{ab}=\frac{P_aGP_b}{g_{ab}}.
\end{equation}
The connected edge set supplied by $B,V$ gives every matrix unit by products along a path. These normalized matrix units are defined at fixed $j$.

\section{Radial product identities}
The Leibniz rule $\partial_r f=f\partial_r+f'$ gives
\begin{equation}
 \mathsf K_+\mathsf K_-=\mathsf K_0(\mathsf K_0-1)-\mathsf C,
 \qquad
 \mathsf K_-\mathsf K_+=\mathsf K_0(\mathsf K_0+1)-\mathsf C.
\end{equation}
At the eigenvalue $n+k_j$, these products are respectively
\begin{equation}
 n(n+2k_j-1),\qquad(n+1)(n+2k_j).
\end{equation}
They give the squared ladder amplitudes in Eqs.~\eqref{eq:radialactions}--\eqref{eq:radialactionminus}; the phase in Eq.~\eqref{eq:unormalized} gives the displayed positive amplitudes.

A first-order factorization of Infeld--Hull type \cite{infeld}, familiar also from higher-dimensional radial oscillator treatments \cite{arik}, provides an independent check of the lowest state. With
\begin{equation}
 a_j=\frac1{\sqrt2}\left(\hbar\frac{d}{dr}+\omega r-\frac{\hbar(j+1)}r\right),
\end{equation}
one finds
\begin{equation}
 a_j^\dagger a_j=H_j-\hbar\omega(j+\tfrac32),\qquad
 a_ja_j^\dagger=H_{j+1}-\hbar\omega(j+\tfrac12).
\end{equation}
The equation $a_ju_{0j}=0$ gives Eq.~\eqref{eq:u0}. This first-order factorization changes the centrifugal parameter, whereas the second-order generators $\mathsf K_\pm$ act at fixed $j$.

\section*{Acknowledgments}
IM was supported by the
Australian Research Council Future Fellowship FT180100099. The work of OOT and \.{I}Y was supported by the Scientific and Technological Research Council of Türkiye (TÜBİTAK) under the 1001 Program, Project No. 123F161.

\section*{Data availability}
No new data were created or analysed in this study.


\begin{thebibliography}{99}
\bibitem{jauch} J. M. Jauch and E. L. Hill, ``On the problem of degeneracy in quantum mechanics,'' Phys. Rev. \textbf{57}, 641--645 (1940). \href{https://doi.org/10.1103/PhysRev.57.641}{doi:10.1103/PhysRev.57.641}.
\bibitem{fradkin} D. M. Fradkin, ``Three-dimensional isotropic harmonic oscillator and SU(3),'' Am. J. Phys. \textbf{33}, 207--211 (1965). \href{https://doi.org/10.1119/1.1971373}{doi:10.1119/1.1971373}.
\bibitem{review} W. Miller Jr., S. Post and P. Winternitz, ``Classical and quantum superintegrability with applications,'' J. Phys. A: Math. Theor. \textbf{46}, 423001 (2013). \href{https://doi.org/10.1088/1751-8113/46/42/423001}{doi:10.1088/1751-8113/46/42/423001}. \href{https://arxiv.org/abs/1309.2694}{arXiv:1309.2694}.
\bibitem{dask91} C. Daskaloyannis, ``Generalized deformed oscillator and nonlinear algebras,'' J. Phys. A: Math. Gen. \textbf{24}, L789--L794 (1991). \href{https://doi.org/10.1088/0305-4470/24/15/001}{doi:10.1088/0305-4470/24/15/001}.
\bibitem{bonatsos} D. Bonatsos, C. Daskaloyannis and K. Kokkotas, ``Deformed oscillator algebras for two-dimensional quantum superintegrable systems,'' Phys. Rev. A \textbf{50}, 3700 (1994). \href{https://doi.org/10.1103/PhysRevA.50.3700}{doi:10.1103/PhysRevA.50.3700}. \href{https://arxiv.org/abs/hep-th/9309088}{arXiv:hep-th/9309088}.
\bibitem{dask} C. Daskaloyannis, ``Quadratic Poisson algebras of two-dimensional classical superintegrable systems and quadratic associative algebras of quantum superintegrable systems,'' J. Math. Phys. \textbf{42}, 1100--1119 (2001). \href{https://doi.org/10.1063/1.1348026}{doi:10.1063/1.1348026}. \href{https://arxiv.org/abs/math-ph/0003017}{arXiv:math-ph/0003017}.
\bibitem{isaac} P. S. Isaac and I. Marquette, ``On realizations of polynomial algebras with three generators via deformed oscillator algebras,'' J. Phys. A: Math. Theor. \textbf{47}, 205203 (2014). \href{https://doi.org/10.1088/1751-8113/47/20/205203}{doi:10.1088/1751-8113/47/20/205203}. \href{https://arxiv.org/abs/1402.3361}{arXiv:1402.3361}.
\bibitem{tanoudis} Y. Tanoudis and C. Daskaloyannis, ``Algebraic calculation of the energy eigenvalues for the nondegenerate three-dimensional Kepler--Coulomb potential,'' SIGMA \textbf{7}, 054 (2011). \href{https://doi.org/10.3842/SIGMA.2011.054}{doi:10.3842/SIGMA.2011.054}. \href{https://arxiv.org/abs/1102.0397}{arXiv:1102.0397}.
\bibitem{chain} Y. Liao, I. Marquette and Y.-Z. Zhang, ``Quantum superintegrable system with a novel chain structure of quadratic algebras,'' J. Phys. A: Math. Theor. \textbf{51}, 255201 (2018). \href{https://doi.org/10.1088/1751-8121/aac111}{doi:10.1088/1751-8121/aac111}. \href{https://arxiv.org/abs/1801.05634}{arXiv:1801.05634}.
\bibitem{dhokervinet84} E. D'Hoker and L. Vinet, ``Supersymmetry of the Pauli equation in the presence of a magnetic monopole,'' Phys. Lett. B \textbf{137}, 72--76 (1984). \href{https://doi.org/10.1016/0370-2693(84)91108-0}{doi:10.1016/0370-2693(84)91108-0}.
\bibitem{dhokervinet85cmp} E. D'Hoker and L. Vinet, ``Dynamical supersymmetry of the magnetic monopole and the $1/r^2$ potential,'' Commun. Math. Phys. \textbf{97}, 391--427 (1985). \href{https://doi.org/10.1007/BF01213405}{doi:10.1007/BF01213405}.
\bibitem{dhokervinet85prl} E. D'Hoker and L. Vinet, ``Constants of motion for a spin-$\tfrac12$ particle in the field of a dyon,'' Phys. Rev. Lett. \textbf{55}, 1043--1046 (1985). \href{https://doi.org/10.1103/PhysRevLett.55.1043}{doi:10.1103/PhysRevLett.55.1043}.
\bibitem{dhokervinet85npb} E. D'Hoker and L. Vinet, ``Spectrum (super-) symmetries of particles in a Coulomb potential,'' Nucl. Phys. B \textbf{260}, 79--102 (1985). \href{https://doi.org/10.1016/0550-3213(85)90311-6}{doi:10.1016/0550-3213(85)90311-6}.
\bibitem{dhokervinet86} E. D'Hoker and L. Vinet, ``Hidden symmetries and accidental degeneracy for a spin-$\tfrac12$ particle in the field of a dyon,'' Lett. Math. Phys. \textbf{12}, 71--76 (1986). \href{https://doi.org/10.1007/BF00400305}{doi:10.1007/BF00400305}.
\bibitem{spin2009} P. Winternitz and \.{I}. Yurdu\c{s}en, ``Integrable and superintegrable systems with spin in three-dimensional Euclidean space,'' J. Phys. A: Math. Theor. \textbf{42}, 385203 (2009). \href{https://doi.org/10.1088/1751-8113/42/38/385203}{doi:10.1088/1751-8113/42/38/385203}. \href{https://arxiv.org/abs/0906.1021}{arXiv:0906.1021}.
\bibitem{spin2012} J.-F. D\'esilets, P. Winternitz and \.{I}. Yurdu\c{s}en, ``Superintegrable systems with spin and second-order integrals of motion,'' J. Phys. A: Math. Theor. \textbf{45}, 475201 (2012). \href{https://doi.org/10.1088/1751-8113/45/47/475201}{doi:10.1088/1751-8113/45/47/475201}. \href{https://arxiv.org/abs/1208.2886}{arXiv:1208.2886}.
\bibitem{scalar} O. O. Tuncer and \.{I}. Yurdu\c{s}en, ``Superintegrability in the interaction of two particles with spin,'' J. Phys. A: Math. Theor. \textbf{58}, 245201 (2025). \href{https://doi.org/10.1088/1751-8121/adde03}{doi:10.1088/1751-8121/adde03}. \href{https://arxiv.org/abs/2410.22321}{arXiv:2410.22321}.
\bibitem{pseudo} F. T\"urkkan, O. O. Tuncer and \.{I}. Yurdu\c{s}en, ``Superintegrability in the interaction of two particles with spin: First-order pseudo-scalar integrals of motion,'' \href{https://arxiv.org/abs/2605.07659v1}{arXiv:2605.07659v1} (2026).
\bibitem{vector} O. O. Tuncer and \.{I}. Yurdu\c{s}en, ``Superintegrable systems of two interacting spin-$\frac12$ particles with first-order vector integrals of motion,'' \href{https://arxiv.org/abs/2607.19155v1}{arXiv:2607.19155v1} (2026).
\bibitem{infeld} L. Infeld and T. E. Hull, ``The factorization method,'' Rev. Mod. Phys. \textbf{23}, 21--68 (1951). \href{https://doi.org/10.1103/RevModPhys.23.21}{doi:10.1103/RevModPhys.23.21}.
\bibitem{arik} M. Ar{\i}k, M. Baykal and A. Baykal, ``Factorization method for $d$-dimensional isotropic harmonic oscillator and the generalized Laguerre polynomials,'' \href{https://arxiv.org/abs/0808.2289v2}{arXiv:0808.2289v2} (2010).
\bibitem{radialsu11} O. Rosas-Ortiz, S. Cruz y Cruz and M. Enriquez, ``SU(1,1) and SU(2) approaches to the radial oscillator: Generalized coherent states and squeezing of variances,'' Ann. Phys. \textbf{373}, 346--373 (2016). \href{https://doi.org/10.1016/j.aop.2016.07.001}{doi:10.1016/j.aop.2016.07.001}. \href{https://arxiv.org/abs/1603.03499}{arXiv:1603.03499}.
\end{thebibliography}
\end{document}